\documentclass[%
superscriptaddress,
 nofootinbib,
 amsmath,amssymb,
 aps,
 prd,
 twocolumn,
 ]{revtex4-2}
\usepackage{graphicx}
\usepackage{dcolumn}
\usepackage{bm}

\usepackage{xcolor}
\usepackage{calc}
\usepackage[export]{adjustbox} 
\usepackage{amsmath,amsfonts,amsthm,amscd,upref,amstext}
\usepackage{romannum} 
\AtBeginDocument{\pagenumbering{arabic}} 
\usepackage{float} 
\usepackage{hyperref} 
\usepackage[normalem]{ulem}

\usepackage{mathtools}

\newcommand{\eqnref}[1]{\eqref{#1}}
\newcommand{\lib}[1]{{\em #1}}
\newcommand{\defeq}{=}

\DeclareMathOperator{\Div}{div}

\newcommand{\R}{\mathbb{R}}
\newcommand{\Lie}{\mathcal{L}}

\newcommand{\LL}{\mathcal{L}}
\newcommand{\spacetime}{\mathcal{N}^4}
\newcommand{\Slice}{\mathcal{S}}
\newcommand{\Surf}{\Sigma}
\newcommand{\Sout}{\Surf_\text{outer}}
\newcommand{\Sin}{\Surf_\text{inner}}
\newcommand{\Sone}{\Surf_A}
\newcommand{\Stwo}{\Surf_B}
\newcommand{\Sonetwo}{\Surf_{A,B}}
\newcommand{\rA}{\tilde r}
\newcommand{\SimA}{Sim1}
\newcommand{\SimB}{Sim2}

\newtheorem{theorem}{Theorem}[section]

\newtheorem{lemma}[theorem]{Lemma}

\newtheorem{rmk}[theorem]{Remark}

\newcommand{\QLM}{\mathrm{QLM}}
\newcommand{\QLE}{\mathrm{QLE}}
\newcommand{\const}{\mathrm{const}}

\newcommand{\MADM}{M^\mathrm{ADM}}
\newcommand{\BYM}{m_\mathrm{BY}}

\begin{document}


\title{Properties of Quasi-local mass in binary black hole mergers}

\author{Daniel Pook-Kolb}
\email{daniel.pook.kolb@aei.mpg.de}
\affiliation{Institute for Mathematics, Astrophysics and Particle Physics, Radboud University, Heyendaalseweg 135, 6525 AJ Nijmegen, The Netherlands}
\affiliation{Albert-Einstein-Institut, Max-Planck-Institut f{\"u}r Gravitationsphysik, Callinstra{\ss}e 38, 30167 Hannover, Germany}
\affiliation{Leibniz Universit{\"a}t Hannover, 30167 Hannover, Germany}

\author{Bowen Zhao}
\email{bowenzhao@bimsa.cn}
\affiliation{Beijing Institute of Mathematical Sciences and Applications, Beijing 101408,
China}

\author{Lars Andersson}
\email{lars.andersson@bimsa.cn}
\affiliation{Beijing Institute of Mathematical Sciences and Applications, Beijing 101408,
China}

\author{Badri Krishnan}
\email{badri.krishnan@ru.nl}
\affiliation{Institute for Mathematics, Astrophysics and Particle Physics, Radboud University, Heyendaalseweg 135, 6525 AJ Nijmegen, The Netherlands}
\affiliation{Albert-Einstein-Institut, Max-Planck-Institut f{\"u}r Gravitationsphysik, Callinstra{\ss}e 38, 30167 Hannover, Germany}
\affiliation{Leibniz Universit{\"a}t Hannover, 30167 Hannover, Germany}

\author{Shing-Tung Yau}
\email{styau@tsinghua.edu.cn}
\affiliation{Beijing Institute of Mathematical Sciences and Applications, Beijing 101408,
China}\affiliation{Yau Mathematical Sciences Center, Tsinghua
University, Beijing 100084, China}

\date{\today}

\begin{abstract}
  Identifying a general quasi-local notion of energy-momentum and
  angular momentum would be an important advance in general relativity with potentially important consequences for mathematical and astrophysical studies in general relativity.  
  In this paper we study a promising approach to this problem first proposed by Wang and Yau in
  2009 based on isometric embeddings of closed surfaces in Minkowski
  space.  We study the properties of the Wang-Yau quasi-local
  mass in high accuracy numerical simulations of the head-on
  collisions of two non-spinning black holes within full general
  relativity.  We discuss the behavior of the Wang-Yau quasi-local
  mass on constant expansion surfaces and we compare its behavior with
  the irreducible mass.  We investigate the time evolution of the
  Wang-Yau Quasi-local mass in numerical examples.  In addition we
  discuss mathematical subtleties in defining the Wang-Yau mass for
  marginally trapped surfaces.
\end{abstract}

\maketitle

\section{Introduction}
\label{sec:intro}

The quasi-local definition of energy-momentum remains one of the major
problems in classical general relativity
\cite{Penrose1982unsolved,Szabados:2009review}.  The goal is to find
appropriate notions of energy-momentum and angular momentum for
finite, extended regions of spacetime.  At spatial infinity and at
null infinity, there are well-established concepts of energy and
angular momentum.  The energy-momentum defined by
Arnowitt-Deser-Misner (ADM) \cite{Arnowitt:1959ah} at spatial infinity
measures the {\em total} energy in a spacetime, and it is conserved
and shown to be positive \cite{Schoen:1979zz,Witten:1981mf}.  The
Bondi energy is measured at null-infinity and satisfies appropriate
balance laws as gravitational radiation carries away energy and
angular momentum \cite{Bondi:1962px}.  Similarly, there are notions of
quasi-local energy and angular momentum and balance laws applicable
for black hole horizons
\cite{Ashtekar:2002ag,Ashtekar:2003hk,Gourgoulhon:2006uc}.  In
contrast, finding suitable analogous definitions for a finitely
extended body or for an arbitrary region in spacetime is still under
active research.

Finding appropriate quasi-local notions of energy momentum and angular
momentum would be desirable for various reasons.  For example, one
might expect that gravitational waves emitted from a given region in
spacetime would carry away energy thus leading to a corresponding
decrease in the quasi-local mass.  Such a link has been shown for the
Bondi mass and for black hole horizons, but is still not available for
general spacetime regions. Once fully understood, it could potentially
allow us to infer detailed properties of dynamical spacetimes in the
strong field region from gravitational wave observations, such as from
the merger of compact objects.  On the mathematical side, it is likely that appropriate quasi-local notions of energy and angular momentum would play an important role in providing a full proof of the Penrose inequality.  Similarly, as we shall discuss in this paper, quasi-local mass also plays an important role in the process of gravitational collapse and black hole formation via the hoop conjecture.

We expect quasi-local mass to be a flux type integral on a closed
space-like 2-surface $\Sigma$ which bounds a space-like hypersurface
$\Omega$. Since $\Omega$ is not unique in the sense of being bounded
by a given $\Sigma$, one would expect that a proper notion of
quasi-local mass should not depend on which specific $\Omega$ is
chosen.  Restricting ourselves to vacuum spacetimes, we can enumerate
some minimal requirements that any viable notion of quasi-local mass
$M_{(\Sigma)}$ should satisfy \cite{liuYau2006}:
\begin{itemize}
\item In flat Minkowski spacetime, $M_{(\Sigma)}$
should vanish.
\item In a curved spacetime, the quasi-local mass should be
  non-negative.
\item In the limit when $\Sigma$ approaches a sphere at spacelike
  infinity on an asymptotically flat slice, or a cross-section of null
  infinity, the quasi-local mass must approach the ADM
  mass or the Bondi mass, respectively.
\item When $\Sigma$ is an apparent horizon, the quasi-local mass must be bounded from below by the irreducible mass of
  $\Sigma$, i.e.\ $\sqrt{A_{\Sigma}/16\pi}$, where $A_\Sigma$ is the area of $\Sigma$.
\end{itemize}
In this work we shall investigate properties of the quasi-local mass
originally proposed by Wang \& Yau \cite{Wang:2008jy}; see also
\cite{liuYau2003}. There are several other proposed definitions of
quasi-local mass, energy-momentum and angular momentum in the
literature.  Some notable ones are due to Bartnik
\cite{Bartnik:1989zz}, Hawking \cite{Hawking:1968qt}, and Penrose
\cite{Penrose:1982wp}; see
\cite{Penrose1982unsolved,Szabados:2009review} for a review.

Based on a variational analysis of the action of General Relativity,
Brown \& York proposed a quasi-local energy arising as a boundary term
in the Hamiltonian \cite{Brown:1992br,LSYork:2007PRD}; see also
\cite{1997GReGr..29..307K,liuYau2003,liuYau2006}.
However, the Brown-York definition
depends explicitly on a choice of spacelike hypersurface $\Omega$ that
is bounded by the two-surface $\Sigma$ under
consideration. Specifically, the mean curvature of $\Sigma$ as
embedded in $\Omega$ appears in the Brown-York definition. Moreover, a
specific choice of unit lapse and zero shift is needed in relating the
Hamiltonian to the Brown-York mass. This rather arbitrary gauge-fixing
is undesirable in general relativity studies. Furthermore, the
Brown-York quasi-local mass can fail to be positive in general except
for the time symmetric case \cite{ShiTam2002}.  On the other hand,
there exist surfaces in Minkowski spacetime with strictly positive
Brown-York mass. These undesirable features are resolved in Wang \&
Yau \cite{wang:2009cmp} by further including momentum information (second fundamental form in the time direction) in their
definition. Indeed, Euclidean space can be regarded as the totally
geodesic space-like hypersurface of zero momentum in Minkowski
spacetime. While Brown \& York defined their reference surface by an
isometric embedding of $\Sigma$ into 3-dimensional Euclidean space
$\R^3$, Wang \& Yau defined their reference surface by an isometric
embedding into Minkowski space $\R^{3,1}$ directly. The positivity proof of Wang-Yau quasilocal mass is given along with the definition \cite{wang:2009cmp}. The new definition is proven to recover the ADM mass at spatial infinity \cite{WangYau:2010spatialinf} and the Bondi mass at null infinity \cite{ChenWangYau:2011nullinf}. Further, the small sphere limit is proven to recover the stress-energy tensor at the limiting point for a spacetime with matter fields and is related to the Bel–Robinson tensor at higher orders for vacuum spacetime. Along the same line, they also give a quasi-local definition for angular momentum and center of mass \cite{ChenWangYau:2015cmp}, which are proven to be supertranslation invariant \cite{chen2021evolution,chen2021supertranslation, chen2022supertranslation}.
We will review the definition of Wang \& Yau quasi-local mass below and compare with Brown
\& York when it is helpful. 

Besides the requirements enumerated above, additional properties would
be desirable when considering the dynamical aspects of general
relativity.  As mentioned earlier, for the Bondi mass at null
infinity, the Bondi mass loss formula shows that gravitational waves
carry away energy, leading to a decrease of the Bondi mass
\cite{Bondi:1962px}.  The flux of gravitational radiation is written
as a surface integral over cross-sections of null infinity, and is
manifestly positive.  Similarly, restricting ourselves to black hole
horizons and marginally trapped surfaces, similar balance laws with
positive fluxes can be shown, leading to a physical process version of
the area increase law
\cite{Ashtekar:2002ag,Ashtekar:2003hk,Gourgoulhon:2006uc}.  Extending
these considerations to a more general quasi-local setting would lead
one to conjecture that the emission (or absorption) of gravitational
radiation from a domain $\Omega$ could be written as a surface flux
integral over $\Sigma$, directly related to the decrease (or increase)
of the quasi-local mass.  At present we do not have a well defined
notion of such fluxes.  As a first step in this direction, in this
work we shall study the time evolution of Wang-Yau quasi-local mass,
henceforth denoted as QLM, in the context of a binary black hole
merger. This question is hard to answer analytically, and we resort
instead to high precision numerical simulations of the full Einstein
equations.  

The plan for this paper is the following.  The basics of the Wang-Yau
QLM and its properties are introduced in Sec.~\ref{sec:basics}.  We
shall consider the head-on collision of two non-spinning black holes
starting with time-symmetric initial data.  The initial data and our
numerical evolution scheme is described in \ref{sec:init-nr}.  Our
numerical implementation for calculating the Wang-Yau QLM, and
numerical convergence, are described in Sec.~\ref{sec:numerics}.  The
numerical results are presented in Sec.~\ref{sec:results} in three
steps.  First, Sec.~\ref{sub:time-symmetric} shows the results in the
initial data, i.e.\ with time symmetry.  As the evolution proceeds,
the later time slices are no longer
time-symmetric. Sec.~\ref{sub:non-time-symmetric} shows results for
non-time-symmetric slices and finally Sec.~\ref{sub:evolution}
presents the time evolution of the QLM and also an exploration of the
hoop conjecture in the context of the formation of the common horizon
in a black hole merger.  In the course of presenting the numerical
results, it will be clear that there are mathematical subtleties in
defining the QLM for a marginally trapped surface.  This will be
clarified mathematically in Sec.~\ref{sec:proofs} and will justify the
various choices made in the numerical work. Finally,
Sec.~\ref{sec:implications} will discuss some implications of our
results and suggestions for future work.

\section{Basic Notions}
\label{sec:basics}

The Wang-Yau quasi-local energy (QLE) associated with a suitable
surface is defined through anchoring the surface intrinsic geometry while comparing
the extrinsic geometry as embedded in the original spacetime $\spacetime$
versus that embedded in the flat Minkowski space $\R^{3,1}$.  Given a
spacelike two-surface $\Surf \subset \spacetime$ with induced metric
$\sigma_{ab}$, let $i_0:\Surf \hookrightarrow \R^{3,1}$ be an
isometric embedding into the Minkowski spacetime.  Fixing a unit,
future-pointing, timelike vector $T_0$ in $\R^{3,1}$, a one-to-one
correspondence between vector fields in $\spacetime$ and those in $\R^{3,1}$ is
built through the `canonical gauge' condition,
\begin{equation}\label{eq:canocinal_gauge}
    \langle H, \bar{e}_4 \rangle = \langle H_0, \check{e}_4 \rangle
    \;,
\end{equation}
where $H$ and $H_0$ are mean curvature vectors of $\Surf\subset \spacetime$ and $i_0(\Surf)\subset \R^{3,1}$,
respectively, and $\langle \,\cdot\, , \,\cdot\, \rangle$ denotes the
corresponding scalar product in $\spacetime$ and $\R^{3,1}$.
The basis vectors $\Bar{e}_\alpha$ of $\spacetime$, $\alpha\in \{0,1,2,3\}$,
and $\Check{e}_\alpha$ of $\R^{3,1}$ are chosen as follows.
Let $\check{e}_3$ be the spacelike unit normal of $i_0(\Surf)$ which is also
perpendicular to $T_0$. Let $\Check{e}_4$ be the future-pointing, timelike
unit normal that is perpendicular to $\Check{e}_3$. Then
$\{\Check{e}_3,\Check{e}_4\}$ forms an orthonormal basis for the normal bundle
of $i_0({\Surf})\subset \R^{3,1}$. The canonical gauge condition
\eqref{eq:canocinal_gauge} picks uniquely a future-pointing, timelike unit
normal of $\Surf$, $\bar{e}_4$. Then $\Bar{e}_3$ is the spacelike normal of
$\Surf$ that combined with $\Bar{e}_4$ gives an orthonormal basis for the
normal bundle of $\Surf \subset \spacetime$.

Given $\tau \in C^\infty(\Surf)$, a generalized mean curvature for $\Surf$ is defined as 
\begin{equation}
    \mathcal{H} = -\sqrt{1+|\nabla \tau|^2}\langle H,\Bar{e}_3\rangle - \alpha_{\Bar{e}_3}(\nabla \tau)
    \;,
\end{equation}
where $\nabla$ denotes the covariant derivative on $\Surf$ associated with $\sigma_{ab}$, $|\nabla \tau|^2 = \sigma^{ab} \nabla_a \tau \nabla_b \tau $ and we write
$\alpha_{\check{e}_3}(\nabla \tau) \defeq (\alpha_{\check{e}_3})_a \nabla^a \tau$.
The connection one-form $\alpha_{\Bar{e}_3}$ associated with the basis
$\{\bar{e}_3,\bar{e}_4\}$ is defined as
\begin{align*}
    \alpha_{\Bar{e}_3}(Y) &= \langle ^{(4)}\nabla_Y \bar{e}_3 \, , \,\bar{e}_4 \rangle
    \;,
\end{align*}
where $Y\in T\Sigma$ and $^{(4)}\nabla$ denotes the covariant derivative in $\spacetime$.
Similarly, one can define $\alpha_{\check{e}_3}$ for the connection one-form associated with $\{\check{e}_3,\check{e}_4\}$ in $\R^{3,1}$ as
\begin{align*}
    \alpha_{\check{e}_3}(Y) = \langle ^{(3,1)}\nabla_Y \check{e}_3 \, , \,\check{e}_4 \rangle
    \;,
\end{align*}
where $^{(3,1)}\nabla_Y$ denotes the covariant derivative in $\R^{3,1}$.
A generalized mean curvature for $i_0(\Surf)$ is defined as
\begin{equation}
    \mathcal{H}_0 = -\sqrt{1+|\nabla \tau|^2}\langle H_0,\check{e}_3\rangle - \alpha_{\check{e}_3}(\nabla \tau)
    \;.
\end{equation}
The Wang-Yau quasi-local energy associated with $\tau$ is then defined as
\begin{equation}
    \QLE (\tau) = \frac{1}{8\pi} \int_\Surf (\mathcal{H}_0-\mathcal{H}) \, dvol_\Surf
   \;.
\end{equation}
When the mean curvature vector $H$ is spacelike, one can use $H=-\langle H, \Bar{e}_4\rangle\, \Bar{e}_4+\langle H, \Bar{e}_3\rangle \,\Bar{e}_3 = p \,\Bar{e}_4 - k \,\Bar{e}_3$ and its conjugate vector $J=k \,\Bar{e}_4 - p\, \Bar{e}_3$ to form an orthonormal basis for the normal bundle $N\Surf$, $\{e_H=-\frac{H}{|H|},e_J=\frac{J}{|H|}\}$.
In terms of this mean curvature vector basis, 
\begin{alignat*}{2}
    &\QLE(\tau) =\\
    &\qquad \frac{1}{8\pi} \int_\Surf \Bigg\{&&\sqrt{1+|\nabla \tau|^2}\cdot (\cosh \theta_0 |H_0| -\cosh \theta |H|)  \\
    &&&- \nabla\tau \cdot \nabla(\theta_0-\theta) -  (\alpha_{H_0}-\alpha_H)(\nabla\tau)\Bigg\}
    \;,
\end{alignat*}
where $\theta$ denotes the hyperbolic angle between $\{\overline{e}_3,\overline{e}_4\}$ and $\{e_H=-\frac{H}{|H|},e_J=\frac{J}{|H|}\}$. Specifically,
\begin{equation}
    \begin{cases}
      \overline{e}_3 = \cosh\theta\, e_H - \sinh\theta\, e_J \\
      \overline{e}_4 = -\sinh\theta \, e_H + \cosh\theta\, e_J
    \end{cases}       
\end{equation}
and similarly for $\theta_0$ in $\R^{3,1}$.

Solving the variational problem of minimizing the QLE with respect to the time function $\tau$,
one gets the Euler-Lagrange equation, called the optimal embedding equation (OEE),
\begin{equation}\label{eq:OEE}
    \nabla_a j^a =0
    \;.
\end{equation}
The minimum value of $\QLE$ is defined to be the Wang-Yau quasi-local mass 
\begin{equation}\label{eq:QLM}
    \QLM = \frac{1}{8\pi} \int_\Surf  \rho + j^a \nabla_a \tau = \frac{1}{8\pi} \int_\Surf  \rho 
    \;,
\end{equation}
where (see (4.4)--(4.5) in \cite{Chen:2010tz}
for details)
\begin{equation}\label{eq:mass_density}
    \rho = \frac{\sqrt{|H_0|^2+\frac{(\Delta\tau)^2}{1+|\nabla \tau|^2}}-\sqrt{|H|^2+\frac{(\Delta\tau)^2}{1+|\nabla \tau|^2}}}{\sqrt{1+|\nabla \tau|^2}}
\end{equation}
and
\begin{equation}\label{eq:j_flux}
    j_a = \rho \nabla_a \tau -\nabla_a \sinh^{-1}\frac{\rho \Delta \tau}{|H_0||H|}-(\alpha_{H_0})_a +(\alpha_{H})_a
    \;.
\end{equation}
The Wang-Yau QLM is defined for any closed spacelike
surface $\Sigma$ whose mean curvature vector is spacelike and where an admissible
solution to the OEE \eqref{eq:OEE} exists (see Definition 5.1 in \cite{wang:2009cmp} for admissible $\tau$).

Note that if $\tau=\const$ is admissible and solves the optimal embedding equation, it must be the global minimum of Wang-Yau quasi-local energy \cite{ChenWangYau:2014cmp}. Substituting $\tau=\const$ to \eqref{eq:QLM}, one sees that the Wang-Yau quasi-local mass reduces to the Liu-Yau mass in this case \cite{liuYau2003,liuYau2006}
\begin{equation}\label{eq:LY-mass}
    \QLM = \frac{1}{8\pi} \int_\Surf |H_0|- |H| \;.
\end{equation}
If further $\Surf$ lies in a totally geodesic slice, 
the Wang-Yau quasi-local mass reduces to the Brown-York mass
\begin{equation}\label{eq:BY-mass}
    \BYM = \frac{1}{8\pi} \int_\Surf k_0- k \;,
\end{equation}
where $k$ is the only nonzero component of $H$ lying in the totally geodesic slice while $k_0$ is the mean curvature vector of $i_0(\Sigma)$ embedded in $\R^3$. Note that $\tau$ only appears through derivatives and we hence use $\tau=0$ and $\tau=\const$ interchangeably.

In black hole spacetimes, there is a particular set of surfaces of
interest---the marginally outer trapped surfaces (MOTSs).  These are
used to study various aspects of black holes quasi-locally via the
framework of isolated and dynamical horizons, respectively describing
black holes in equilibrium and in dynamical situations (see
e.g. \cite{Ashtekar:2004cn,Booth:2005qc,Gourgoulhon:2005ng,Hayward:2004fz,Jaramillo:2011zw}).
In any given spacelike slice $\Slice \subset \spacetime$, the
outermost MOTS is called the apparent horizon.  Given a MOTS on a
spacelike slice, the results of Andersson et
al. \cite{Andersson:2005gq,Andersson:2007fh,andersson2009area} show
the conditions under which it evolves smoothly.  It is shown that when
a MOTS is stable under outward deformations, then it will evolve
smoothly.  Recent numerical work has applied and further explored the
stability of MOTSs and its implications for the time evolution
\cite{Pook-Kolb:2019iao,Pook-Kolb:2019ssg,Pook-Kolb:2020zhm,Pook-Kolb:2021gsh,Pook-Kolb:2021jpd,Booth:2021sow}.

A MOTS is defined as follows. Let $\Surf \subset \spacetime$ be a closed spacelike
surface and let $\ell^+$, $\ell^-$ be two future directed null normal fields
on $\Surf$ taken to point outward and inward, respectively.
We fix the cross normalization by $\langle \ell^+, \ell^- \rangle = -2$, which
still leaves a remaining freedom to scale
\begin{align}\label{eq:ell-scaling}
    \ell^+ &\to f\ell^+ \;,
    &\ell^- &\to \frac{1}{f}\ell^-
\end{align}
for any positive function $f$.  The outward and inward null
expansions, denoted $\Theta_+$ and $\Theta_-$ respectively, are
defined as
\begin{equation}\label{eq:Theta}
    \Theta_{\pm} \defeq \sigma^{\alpha\beta}\  {}^{(4)}\nabla_\alpha \ell^{\pm}_\beta
    \;,
\end{equation}
where $\sigma^{\alpha\beta} \defeq \sigma^{ab} \pi^\alpha_a \pi^\beta_b$
with $\pi^\alpha_a$ the projection onto the tangent bundle of $\Surf$.
Then, $\Surf$ is called a
{\em marginally outer trapped surface} (MOTS) if $\Theta_{+} = 0$, an
{\em outer trapped surface} if $\Theta_{+} < 0$ and an
{\em outer untrapped surface} if $\Theta_{+} > 0$.
A {\em marginally trapped surface} (MTS) is a MOTS with $\Theta_{-} < 0$.
Note that although $\Theta_{+} \to f\Theta_{+}$ under \eqnref{eq:ell-scaling},
the signs of $\Theta_{\pm}$ are invariant and so these definitions are not
affected.

We will often consider families of surfaces $\Surf_s$ with constant
$\Theta_{+} = s \in \R$, called {\em constant expansion surfaces} (CESs).
These do depend on the choice of $\ell^\pm$, which we fix uniquely using the
spacelike slice $\Slice \subset \spacetime$ within which the family $\Surf_s$ is
constructed.
Concretely, let $\Surf_s \subset \Slice$,
$v$ the spacelike outward unit normal of $\Surf_s$ in $\Slice$
and $u$ the future timelike unit normal on $\Slice$ in $\spacetime$.
Then, we choose
\begin{equation}\label{eq:ell-choice}
    \ell^\pm = u \pm v \;.
\end{equation}
In terms of the null expansions $\Theta_{\pm}$, the mean curvature vector $H$
and its conjugate vector $J$ can be expressed as
\begin{align}\label{eq:HJ-via-Theta}
    H &= \frac{\Theta_{+} \ell^- + \Theta_{-} \ell^+}{2} \;,&
    J &= \frac{\Theta_{+} \ell^- - \Theta_{-} \ell^+}{2} \;.
\end{align}
Since $\langle H, H \rangle = - \Theta_{+} \Theta_{-}$, the mean curvature
vector becomes a null vector on a MOTS.
If in addition the slice $\Slice$ is time symmetric, i.e.\ its second fundamental
form vanishes, then $\sigma^{\alpha\beta}\ {}^{(4)}\nabla_\alpha u_\beta = 0$
and thus
$\Theta_{+} = - \Theta_{-}$
implying that $H$ and $J$ both vanish on a MOTS.

To characterize the quasi-local mass of a black hole region, we take the QLM of apparent horizons. However, the current definition of the Wang-Yau quasi-local mass assumes the surface mean curvature vector to be spacelike and hence does not apply to MOTSs. Therefore, one of our goals is to extend the definition of the
Wang-Yau QLM \eqnref{eq:QLM} to a MOTS in time symmetry and to an MTS without
time symmetry.
Limiting ourselves to the case of axisymmetry and no angular momentum, we will
argue that a suitable extension is
\begin{align}\label{eq:numerical-QLM-on-MOTS}
    \QLM &= \frac{1}{8\pi} \int_\Surf |H_0| &
    \text{with } \tau = \const
    \;.
\end{align}
With this extension, we can then investigate the time evolution of QLM
during black hole collisions.  As noted above if $\Sigma$ lies in a
totally geodesic slice, e.g.\ in the moment of time-symmetry, Wang-Yau
QLM reduces to Brown-York mass. In this case, the above extension
simply reduces to $\QLM=\frac{1}{8\pi}\int k_0$, which is what one
would expect for the Brown-York limit at minimal surfaces. Further extension of QLM to surfaces of timelike mean curvature vector, e.g. trapped surfaces inside event horizons, is certainly of great interest and will be studied elsewhere.

\section{Initial Data and numerical evolution}
\label{sec:init-nr}

We use Brill-Lindquist initial data \cite{BrillLindquist1963:PR}, which solves the
constraint equations of General Relativity with vanishing extrinsic curvature
and vanishing scalar curvature, i.e.\ a time-symmetric slice in vacuum spacetime.
The Riemannian three-metric is defined on $\R^3 \setminus \{x_1, \ldots, x_n\}$
with $n+1$ asymptotically flat ends, one at $\|x\| \to \infty$ and $n$ at the
punctures $x_i$.
We restrict ourselves to the case $n=2$, which describes a two-black-hole
configuration.
The three-metric can then be written as
\begin{equation}\label{eq:BL}
    h_{ij} = \Phi^4 \delta_{ij} \;,
\end{equation}
where $\delta_{ij}$ is the flat metric and the conformal factor is
\begin{equation}\label{eq:Phi}
    \Phi = 1 + \frac{m_A}{2|x-x_A|} + \frac{m_B}{2|x-x_B|}
    \;.
\end{equation}
We take the two punctures to be located on the $z$-axis at coordinates
$x_{A,B} = (0, 0, \pm d/2)$, respectively.
The three ends at $\|x\|\to\infty$, $x_A$ and $x_B$, respectively,
have ADM-masses
\begin{align}\label{eq:BL-ADM}
    \MADM &= m_A + m_B \;,\\
    \MADM_A &= m_A + \frac{m_A\, m_B}{2d} \;, \\
    \MADM_B &= m_B + \frac{m_A\, m_B}{2d} \;.
\end{align}

\begin{figure}
    \centering
    \includegraphics[width=0.9\linewidth]{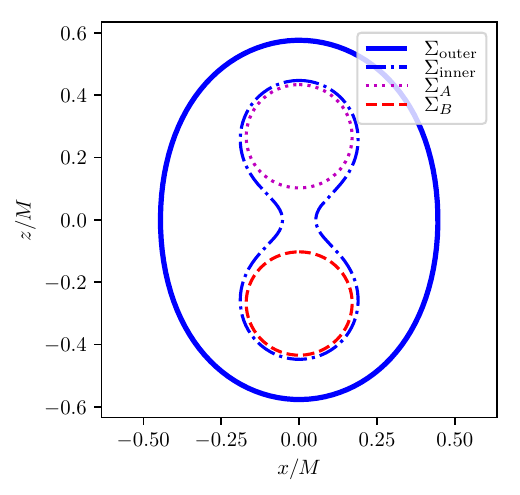}%
    \caption[]{\label{fig:MOTSs}%
        Common MOTS $\Sout$, inner common MOTS $\Sin$
        and the two individual MOTSs $\Sone$ and $\Stwo$ in Brill-Lindquist
        initial data with parameters $m_A=m_B=d=1/2$.
    }
\end{figure}

For sufficiently large $d$, the slice $\Slice$ contains two separate
black holes, each surrounded by a stable MOTS that contains either $x_A$ or $x_B$.
We shall call these the {\em individual MOTSs} $\Sone$ and $\Stwo$,
respectively.
If $d$ becomes small enough, there exists a stable common MOTS $\Sout$ surrounding
$\Sonetwo$.
In fact, as $d$ passes through the value at which $\Sout$ appears, it is found
that an unstable MOTS $\Sin$ forms together with $\Sout$ and ``moves'' inward as
$d$ is decreased. This is discussed in more detail elsewhere
\cite{Pook-Kolb:2018igu}.\footnote{%
    There is a large number of additional MOTSs in these data
    \cite{Pook-Kolb:2021gsh,Booth:2021sow,Pook-Kolb:2021jpd}, which are all
    found to be unstable.
    We will hence not discuss these surfaces in the present work.
}
The two common and two individual MOTSs for an equal mass configuration
are shown in Fig.~\ref{fig:MOTSs}.

The numerical data for this initial slice are generated by
the \lib{TwoPunctures} \cite{Ansorg:2004ds} thorn of the
\lib{Einstein Toolkit} \cite{Loffler:2011ay,EinsteinToolkit:web}.
These data are evolved in time using an axisymmetric version of
\lib{McLachlan} \cite{Brown:2008sb}, which in turn uses \lib{Kranc}
\cite{Husa:2004ip,Kranc:web} for generating C++ code.
This uses the BSSN formulation of the Einstein equations with gauge conditions
chosen as the so-called $1+\log$ slicing and a $\Gamma$-driver shift condition
\cite{Alcubierre:2000xu, Alcubierre:2002kk}.
More details about our numerical simulation setup, including a convergence
analysis, are described in \cite{Pook-Kolb:2019ssg}.

Our analysis is based on two simulations, both starting from BL data.
The first, referred to as {\SimA}, uses initial data with
$m_B/m_A = 2, d=0.9$ and the second,
simulation {\SimB}, uses $m_B/m_A = 1.6, d=1$.
Both simulations were performed with different
spatial grid resolutions to check the accuracy of our calculations.
Results shown for {\SimA} use a resolution $1/\Delta x = 720$, which was
evolved until simulation time $t_f = 6$.
For {\SimB}, we used a lower resolution of $1/\Delta x = 312$ to extend the
evolution up to time $t_f = 38$.

The MOTSs and CESs are numerically found with high accuracy both in the
analytical initial data as well as in slices produced by the
\lib{Einstein Toolkit} using the method in
\cite{Pook-Kolb:2018igu,pook_kolb_daniel_2021_4687700}.

In general, the problem of locating a surface $\Surf_s$ with expansion
$\Theta_{+} = s$ may have many solutions within a given slice $\Slice$.
For $s=0$, this corresponds to the different MOTSs in $\Slice$.
By choosing suitable initial guesses for the numerical search, we can easily
select which particular MOTS to find.
As mentioned above, we focus here on the three stable MOTSs $\Sout$, $\Sone$
and $\Stwo$, interpreted as the horizon of the merger remnant and the smaller
and larger (in case of unequal masses) individual black holes.
Choosing one of these MOTSs as initial guess, we construct CESs for $s$ close
to zero.
Families of $\Surf_s$ are then built by taking small steps in $s$, each time
using the previous CES as initial guess for the next.
CESs far from $\Sout$ in the nearly flat region of $\Slice$ are close to being
spherical in our coordinates and so we can use coordinate spheres as
initial guesses in this case.

\section{Numerical Method for Evaluating the Quasi-local Mass}
\label{sec:numerics}

The strategy to solve the optimal embedding equation $\nabla_a j^a=0$ is to
consider $\nabla_a j^a$ as a nonlinear operator $\LL$
acting on $\tau$, linearize that operator $\LL$ and solve the linear problem
multiple times, each time taking a small step towards a solution of the full
nonlinear problem.
This is also called the {\em Newton-Kantorovich} method
\cite[Appendix~C]{BoydBook}.

Analytically linearizing $\LL$ requires determining the explicit dependency of
$\nabla_a j^a$ on $\tau$.
We make use of axisymmetry, i.e.\ assuming the surface $\Surf$, the embedding $i_0$ and $\tau$ are all axisymmetric, to simplify calculations. The time function
$\tau$ then depends only on one parameter, say $\theta$, increasing from one
pole of $\Surf$ to the other, and the embedding $i_0$ can be expressed in
terms of the intrinsic metric $\sigma_{ab}$ and $\tau'$, where
$\tau' \defeq \frac{d\tau}{d\theta}$.
Explicitly, for coordinates $\{y^1 = \theta, y^2 = \phi\}$ on $\Surf$,
$0 < \theta < \pi$ and $0 < \phi < 2\pi$,
an axisymmetric ansatz for the embedding is
\begin{equation}\label{eq:embedding-ansatz}
    i_0(\theta, \phi) = \begin{pmatrix}
        \tau(\theta) \\ R(\theta) \cos\phi \\ R(\theta) \sin\phi \\ Z(\theta)
    \end{pmatrix}
    \;.
\end{equation}
Writing the two-metric as
$d\sigma^2 = P^2\, d\theta^2 + Q^2\sin^2\theta\, d\phi^2$,
we then have
\begin{align}\label{eq:embedding-sol}
    R^2 &= Q^2 \sin^2\theta \;, &
    Z'^2 &= P^2 - V'^2 + \tau'^2 \;,
\end{align}
where $V \defeq Q\sin\theta$.
Using this, we calculate $\langle H_0, H_0 \rangle = k_0^2 - p_0^2$ via
\begin{align}\label{eq:k0p0}
    k_0 &= \frac{VV''Z' - P^2 Z' + VV'Z''}{\sqrt{P^2+\tau'^2} P^2 V} \;,\\
    p_0 &= \frac{P'V\tau' - P(V\tau''+V'\tau')}{\sqrt{P^2+\tau'^2} P^2 V} \;.
\end{align}
Furthermore,
\begin{equation}\label{eq:alpH0}
    (\alpha_{H_0})_\theta = \frac{k_0p_0' - p_0k_0'}{|H_0|^2}
        + \tau' \frac{V'Z''-V''Z'}{P(P^2+\tau'^2)}
        \;.
\end{equation}

The respective terms in curved space are calculated differently using the
null expansions, which we have in highly accurate form from the MOTS and CES
finding process.
In addition to \eqnref{eq:HJ-via-Theta}, we use
\begin{equation}\label{eq:alpH-TpTm}
    (\alpha_H)_a = \frac{1}{2} \left(
        \frac{\Theta^+_{,a}}{\Theta_+} - \frac{\Theta^-_{,a}}{\Theta_-}
        - \ell_+^\mu \pi^\nu_a {}^{(4)}\nabla_\nu \ell^-_\mu
    \right)
    \;,
\end{equation}
where $\Theta^{\pm}_{,a} = \frac{\partial \Theta_{\pm}}{\partial y^a}$.
We remark that $\Theta_{\pm}$ contains first derivatives of the 3-metric,
which means that $\nabla_a j^a$ contains third derivatives.
In order to get numerically accurate results,
we expand $\Theta_{\pm}$ into a set of basis functions and
differentiate these directly.

\begin{figure}
    \centering
    \includegraphics[width=1.0\linewidth,valign=m]{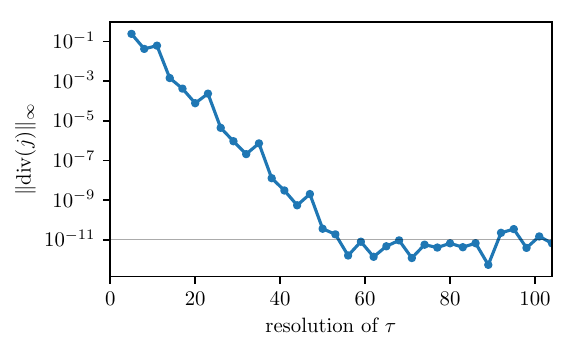}%
    \caption[]{\label{fig:convergence}%
        Maximum residual of the OEE \eqnref{eq:OEE} as the resolution of our
        representation of $\tau$ is increased.
        The error drops exponentially up to reaching a numerical roundoff at
        about $10^{-11}$.
        The case shown here is for a CES with $\Theta_{+} \approx 0.114$
        in a non-time-symmetric slice.
    }
\end{figure}

To linearize the operator $\LL \tau \defeq \nabla_a j^a$, the above
expressions are first inserted in turn into \eqnref{eq:mass_density},
\eqnref{eq:j_flux} and $\nabla_a j^a$.
Afterwards, we use
\lib{SymPy} \cite{meurer2017sympy}
to symbolically differentiate $\LL$ with respect to $\tau'$, $\tau''$,
$\tau^{(3)}$ and $\tau^{(4)}$.
In the end, the linearized operator we implement into our numerical code is of
the form
\begin{equation}\label{eq:linearized-op}
    (\delta\LL)\Delta
        = \sum_{n=1}^4 (\delta_{\tau^{(n)}} \LL) \frac{\partial^n \Delta}{\partial\theta^n}
        \;,
\end{equation}
where $\Delta$ is a scalar function on $\Surf$.
Starting with an initial guess $\tau_0$, usually $\tau_0 = 0$, we perform
steps $\tau_{i+1} = \tau_i + \Delta_i$, where $\Delta_i$ solves the linear
equation
\begin{equation}\label{eq:lienar-eq}
    (\delta\LL)\Delta_i = - \LL \tau_i
    \;,
\end{equation}
which we solve using a pseudospectral method.
In most cases, it took between $5$ and $15$ steps to converge up to
numerical roundoff.
Fig.~\ref{fig:convergence} shows that the residual of the OEE \eqnref{eq:OEE}
decreases exponentially with the resolution of $\tau$, where the resolution is
the number of basis functions used for the finite representation of $\tau$.

\section{Numerical Results}
\label{sec:results}

\subsection{The QLM in time-symmetric initial data}
\label{sub:time-symmetric}

For a time-symmetric slice $\Slice$, the mean curvature vector $H$ of $\Surf$
lies in $\Slice$. 
A MOTS therefore coincides with a minimal surface
($k=-\langle H, v\rangle = 0$, where, as before, $v$ is the
outward unit normal of $\Surf$ in $\Slice$).
Moreover, in time symmetry we have $\alpha_H=-\langle {}^{(4)}\nabla e_J, e_H \rangle \equiv 0$ since $e_H$ lies in $\Slice$ while $e_J$ is the normal to the totally geodesic slice $\Slice$.
And for $\tau = \const$, $i_0(\Sigma)\subset \R^3$, by a similar argument, $\alpha_{H_0} \equiv 0$.
Thus $\tau=\const$ trivially solves the
optimal embedding equation \eqnref{eq:OEE}.
Furthermore, $\tau=\const$ is known to be the global minimum of the QLE provided that it solves
the optimal embedding equation \cite{ChenWangYau:2014cmp}. The Wang-Yau
quasi-local mass reduces to the Brown-York mass $\BYM$ for any surface $\Surf$ in a moment of time symmetry.

\subsubsection{On the monotonicity along geometric flows}
\label{sub:monotonicity}

It is well known that for some cases, the Brown-York mass exhibits a monotonically decreasing behavior \cite{martinez:1994PRD}.
As an example, consider a Schwarzschild black hole.
In a time-symmetric slice, with the metric in isotropic coordinates
\begin{equation}\label{eq:SchwarzschildMetric}
    ds^2 = \left(1+\frac{m}{2r}\right)^4 \big(dr^2+ r^2 d\Omega^2\big)
    \;,
\end{equation}
the horizon lies at $r=m/2$.
One can show that $$\BYM\left(r=\frac{m}{2}\right) = 2m > \BYM(r=\infty) = m \;.$$
This is interpreted as negative gravitational field energy bringing down
$\BYM$ as the surface approaches infinity \cite{Szabados:2009review}.

In fact, this monotonicity property could be shown more precisely.
Consider two mean convex surfaces in $\Slice$, $\Surf_i=\partial \Omega_i$, $i\in\{1,2\}$, with $\Omega_1\subset \Omega_2$ and suppose there exists a geometric flow 
$$\frac{dF}{dt} = f v, \quad f>0$$
from $F(t_1)=\Surf_1$ to $F(t_2)=\Surf_2$.
Then it is proven \cite[Corollary~3.3]{MiaoShiTam:2010CMP} that
\begin{multline}\label{eq:BY-gravitational-energy}
    \BYM(\Surf_2)-\BYM(\Surf_1)\\ = \frac{1}{16\pi} \int_{\Omega_2\setminus\Omega_1} R + |B_0 -B|^2 - (k_0-k)^2
        \;,
\end{multline}
where $B$ and $B_0$ are the second fundamental forms of
$\Surf_t \subset \Slice$ and of $i_0(\Surf_t)\subset \R^3$,
respectively.  For a time-symmetric slice $\Slice$, the scalar curvature
$R=2 T_{00}$ and hence the $\int R$ term can be interpreted as matter
contribution, which in our case vanishes. The remaining
$\int |B_0 -B|^2 - (k_0-k)^2 $ term is then supposed to characterize
the pure gravitational field energy. Note that although the integrand
$|B_0-B|^2-(k_0-k)^2$ clearly depends on the foliation $\Sigma_t$, the
total integral does not. 
The work of Huisken \& Yau \cite{HuiskenYau:1996Inventions} and later improvements \cite{QingTian:2007uniqueness,metzger2007foliations, huang2010foliations,nerz2015foliations,EichmairKoerber:2022CMCfoliations} show that ends of an asymptotically flat Riemannian 3-manifold with positive ADM mass and non-negative scalar curvature admit a unique canonical foliation through stable constant mean curvature (CMC) surfaces. The geometric flow is assured in this case with $\Sigma_t$ being CMC surfaces.
Nonetheless, assume \eqref{eq:BY-gravitational-energy} holds true (in our case true numerically, see below), one can discuss the sign
  of the gravitation field energy term. For simplicity, take an
  orthonormal basis for $\Surf_t$, $\{e_1,e_2\}$, such that
  $\sigma_{ab}=\delta_{ab}$. This basis is also an orthonormal basis
  for the isometric embedding $i_0(\Surf_t)$. Then in this basis,
\begin{equation}
    |B_0-B|^2-(k_0-k)^2= -2 \,\mathrm{det} (B_0-B)
    \;,
\end{equation}
where $B_0-B=D^{\R^3}v_0 - D v:T\Sigma \to T\Sigma$ is regarded as a linear map in this ON basis, with $D^{\R^3}$ and $D$ denoting covariant differentiation in $\R^3$ and in Riemannian $\Slice$, respectively.
If $B_0-B$ can be chosen to be orientation preserving throughout a foliation $\Sigma_t$, then $\int_\Sigma |B_0-B|^2-(k_0-k)^2 <0$, i.e.\ it indicates negative gravitational field energy. Then $\BYM$ would be monotonically decreasing as the surfaces approach infinity. We suspect this is generally true for at least mean convex $\Surf$ but a proof is missing at the moment.

We examine the above equality \eqnref{eq:BY-gravitational-energy}
numerically.  Fig.~\ref{fig:CESs-QLM-AH-far} shows the QLM for CESs
with $\Theta_{+} \geq 0$. Outside the apparent horizon, the QLM
decreases monotonically with increasing distance to the horizon,
whereas the expansion $\Theta_+=k$ increases from $0$ at first (right
panel, upper part) and then drops back to $0$ at infinity (right
panel, lower part). The balance \eqref{eq:BY-gravitational-energy} is
also verified (see Fig.~\ref{fig:Miao-Shi-Tam}) although we cannot
prove for now that there exists a geometric flow among constant
expansion (mean curvature) surfaces in our case.

\begin{figure}
    \centering
    \includegraphics[width=1.0\linewidth,valign=m]{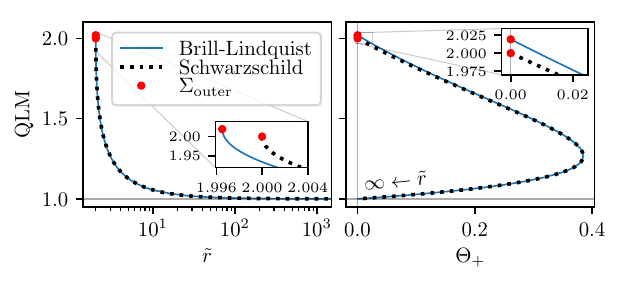}%
    \caption[]{\label{fig:CESs-QLM-AH-far}%
        QLM calculated for a family of CESs in BL data with
        $m_A = m_B = d = 1/2$,
        going from spatial infinity ($\QLM_\infty = 1$) to the common horizon
        $\Sout$ ($\QLM_\text{outer} \approx 2.019209822$).
        The dotted line shows the QLM for constant radius surfaces in the
        time-symmetric Schwarzschild slice for comparison.
        The left panel shows the QLM as function of the area radius $\rA$,
        defined by $4\pi \rA^2 = A$, where $A$ is the area, and the right panel as
        function of the expansion $\Theta_{+}$.
    }
\end{figure}

\begin{figure}
    \centering
    \includegraphics[width=1\linewidth,valign=m]{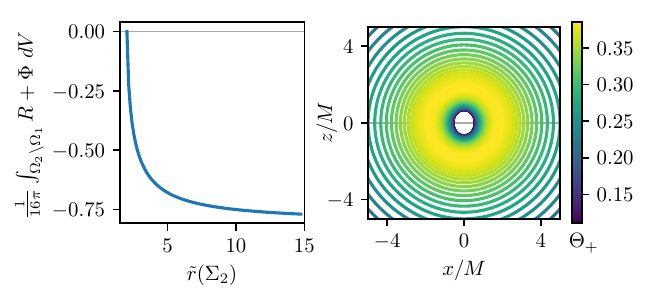}%
    \caption[]{\label{fig:Miao-Shi-Tam}%
        Numerical evaluation of \eqnref{eq:BY-gravitational-energy}.
        The left panel shows the numerical integral, where
        $\Phi \defeq |B_0 -B|^2 - (k_0-k)^2$.
        It agrees with the separately computed difference
        $\QLM(\Sigma_2)-\QLM(\Sigma_1)$ to within about $10^{-7}$.
        The $x$-axis represents the area radius of the outer surface
        $\Surf_2$, which is varied from agreeing with $\Surf_1$ to an almost
        spherical CES of area radius $\rA = 15$.
        The inner surface $\Surf_1$ is a CES outside the apparent horizon with
        expansion $\Theta_{+} = 0.1$.
        The right panel shows part of the CES family integrated over.
        The plots were produced with a Brill-Lindquist setup with
        $m_A = m_B = d = 1/2$.
    }
\end{figure}

\subsubsection{Outer trapped surfaces in time symmetry}
\label{sub:trapped}

\begin{figure}
    \centering
    \includegraphics[width=1.0\linewidth,valign=m]{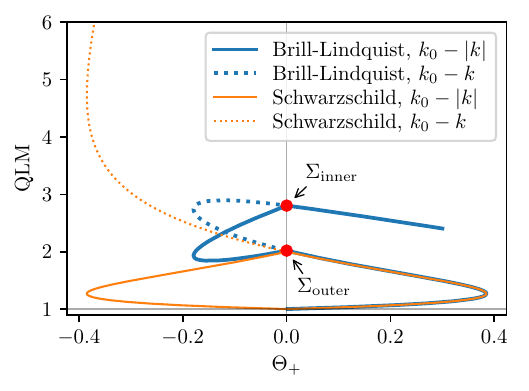}%
    \caption[]{\label{fig:BY-vs-WY}%
        Same as Fig.~\ref{fig:CESs-QLM-AH-far}, but continuing the families
        to the inside of the apparent horizon in both the BL data (blue) and in a
        Schwarzschild slice (orange).
        We show a comparison of the
        Brown-York mass \eqnref{eq:BY-mass}
        calculated either via $\frac{1}{8\pi}\int_\Surf k_0 - |k|$ (solid lines)
        or via $\frac{1}{8\pi}\int_\Surf k_0 - k$ (dotted lines).
        For $\Theta_{+} \geq 0$, the two definitions agree.
        In the Brill-Lindquist case, the CES family interpolates between the
        inner common MOTS $\Sin$ and the apparent horizon $\Sout$ via surfaces
        with $\Theta_{+} < 0$.
        However, these latter surfaces intersect each other (Fig.~\ref{fig:CESs-AH-inner}).
        Monotonicity in this regime can therefore not be expected. See text for discussion.
    }
\end{figure}

In time symmetry, we have $\Theta_{+} = - \Theta_{-}$ and, since
\begin{equation}\label{eq:k-of-Tpm}
    k = -\langle H, v\rangle = \frac{\Theta_+-\Theta_-}{2}
    \;,
\end{equation}
$k < 0$ everywhere for outer trapped surfaces ($\Theta_+<0$).
However, both the Wang-Yau QLE and the Brown-York mass implicitly make the
assumption that $k > 0$ and hence do not apply to such surfaces.
We argue that for $k=-|H|<0$ in the time-symmetric case, the same formula works with $k$ replaced by its absolute value.

We first exemplify this with a time-symmetric slice of the Schwarzschild
metric \eqref{eq:SchwarzschildMetric}.
Recall that there is an isometric inversion $r\to \frac{(m/2)^2}{r}$ that sends surfaces inside the horizon to the outside, reversing the sign of $H$ while preserving $H_0$ ($H_0$ only depends on the surface metric). Taking the absolute value of $k$, one has for a surface of constant radius
\begin{align}\label{eq:SchwarzschildQLM}
    \QLM&=\frac{1}{8\pi} \int k_0 - |k|
    =\begin{cases}
        m(1+\frac{m}{2r}) & r \geq \frac{m}{2}\\
        2r(1+\frac{m}{2r}) & r \leq \frac{m}{2} 
    \end{cases}
    \;.
\end{align}
This is shown in Fig.~\ref{fig:BY-vs-WY} (orange), where \eqnref{eq:SchwarzschildQLM}
is plotted together with the case of no absolute value taken (dotted),
and compared with the analogous Brill-Lindquist case (blue).
Note that $r$ denotes the isotropic radial coordinate.
The QLM attains its maximum value $2m$ at the horizon while
both the $r\to \infty$ and $r\to 0$ limit yield $\QLM=m$, consistent with the interpretation of a time-symmetric slice in Schwarzschild as a wormhole connecting two identical, asymptotically flat regions. More specifically, for $\Theta_{+} < 0$, the surface lies in the other asymptotic region, where its {\em outward} normal is $-v$. If we use $k < 0$ as-is, then the BY-mass remains monotonic and diverges as $\|x-x_i\|\to0$ (i.e.\ approaching the other end $x_i$).
In summary, a naive extension of Brown-York mass into the $k<0$ region results in a smooth QLM profile while taking absolute value yields a ``kink'' at the horizon. In this case, a non-smooth QLM profile ($k_0-|k|$) is clearly more natural than a smooth one ($k_0-k$) and one might in general expect some non-analytic behavior of QLM around MOTSs or horizons.

\begin{figure}
    \centering
    \includegraphics[width=1.0\linewidth,valign=m]{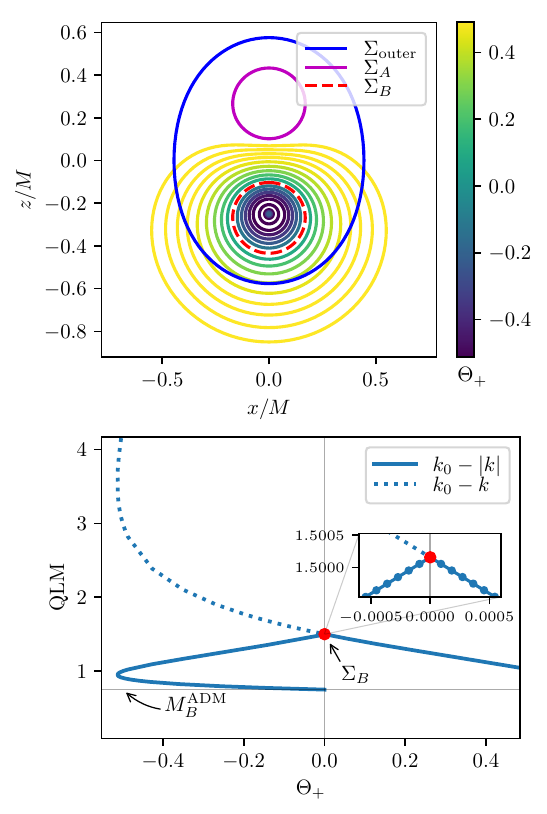}%
    \caption[]{\label{fig:CESs-QLM-S2}%
        Upper panel: CESs near the individual MOTS $\Stwo$ in BL data with
        $m_A = m_B = d = 1/2$.
        Lower panel: QLM calculated for these CESs. The inset shows a
        close-up near $\Theta_{+} = 0$.
        As in Fig.~\ref{fig:BY-vs-WY}, the dotted line shows
        \eqref{eq:BY-mass} calculated via
        $\frac{1}{8\pi}\int_\Surf k_0 - k$.
        The QLM asymptotes to the ADM mass $3/4$ of the end at $x_B$
        as the family approaches the puncture $x_B$.
    }
\end{figure}

For multiple black holes, taking the absolute value of $k$ yields results consistent with \cite{BrillLindquist1963:PR}. That is, for large spheres one recovers the ADM mass as expected while for small spheres approaching each puncture $x_i$, an asymptotic expansion yields the ADM mass associated to each puncture \eqnref{eq:BL-ADM}.
Numerical calculation for CES around $\Sout$ is also included in Fig.\ \ref{fig:BY-vs-WY} (blue) and reveals a similar behavior as in the Schwarzschild case up to $\Theta_+\approx -0.2$ where a turn-over happens in the multiple holes BL data (see below for discussion).
Numerical calculation for CES toward each puncture, i.e.\ around $\Sonetwo$, is shown in Fig.~\ref{fig:CESs-QLM-S2} and again reveals a similar behavior as in the Schwarzschild case.

We emphasize that in Brill-Lindquist data, only CES families outside $\Sout$ and inside each individual $\Sonetwo$ are comparable with the Schwarzschild case. 
The region bounded by $\Sout$ surrounding all punctures and $\Sonetwo$ surrounding each individual puncture does not have a direct correspondence in the Schwarzschild case.
The CES family interpolates between $\Sout$ and the
unstable common MOTS $\Sin$ and does not approach either of the two asymptotic
ends $x_{A,B}$ (Fig.~\ref{fig:CESs-AH-inner}).
Moreover, in this region, constant expansion or constant mean curvature (CMC) surfaces fail to foliate space and monotonicity of the QLM indeed fails here. These explain peculiar features around $\Sin$ in BL data seen in Fig.~\ref{fig:BY-vs-WY}. We remark that choosing other families of foliating surfaces such as coordinate spheres yields qualitatively similar conclusions.
The choice of families of CESs, which in time symmetry are constant mean
curvature surfaces, allows us to get arbitrarily close to a MOTS, which is a CES itself.

\begin{figure}
    \centering
    \includegraphics[width=1.0\linewidth,valign=m]{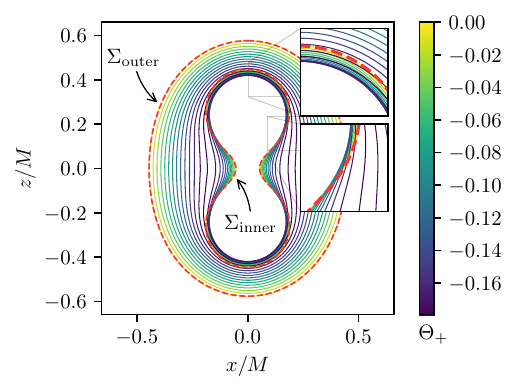}%
    \caption[]{\label{fig:CESs-AH-inner}%
        Family of CESs interpolating between $\Sout$ and $\Sin$ in
        BL data with parameters
        $m_A = m_B = d = 1/2$.
        Close to $\Sin$, members of this family intersect each other and hence fails to foliate the space in this interior region.
    }
\end{figure}

Our extension confirms that QLM at the MOTS
\eqref{eq:numerical-QLM-on-MOTS} should be
\begin{align}\label{eq:QLM-MOTS-Minkowski}
    \QLM = \frac{1}{8\pi}\int H_0 \ge 2\sqrt{\frac{|\Surf|}{16\pi }}
\end{align}
where $|\Surf|$ denotes the area of $\Surf$ and the Minkowski inequality is invoked. This is already assumed in earlier studies \cite{Szabados:2009review}.

\subsection{QLM in non-time-symmetric slices}
\label{sub:non-time-symmetric}

As one numerically evolves the time-symmetric initial data, the slices
$\Slice$ become non-time-symmetric and the mean curvature vector of a
surface $\Surf \subset \Slice$ may acquire a timelike component.
The Wang-Yau quasi-local mass will then in general differ from the
Brown-York mass.

\begin{figure}\centering
    \includegraphics[width=1.0\linewidth,valign=t]{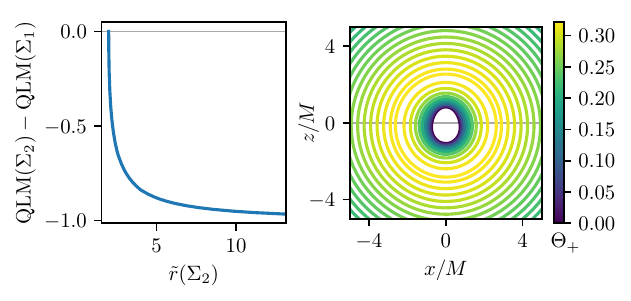}%
    \caption[]{\label{fig:BL31-QLM-MOTS-monotonic}%
        Monotonicity of the QLM along a CES family going outward from
        $\Surf_1 = \Sout$ at time $t=2.5\,M$ in simulation {\SimA}.
        The area radius $\rA(\Surf)$ increases monotonically along this
        family.
        The left panel shows the difference of the QLM between $\Sout$ and
        $\Surf_2$ as $\Surf_2$ is moved outward.
        The right panel shows the region around $\Sout$ foliated by the CES
        family.
    }
\end{figure}

This has various consequences, one being that the monotonicity of the QLM
along geometric flows is not guaranteed by \eqref{eq:BY-gravitational-energy}
anymore, although we numerically find monotonicity remains true, as can be seen in
Fig.~\ref{fig:BL31-QLM-MOTS-monotonic}.
An analytic generalization of \eqref{eq:BY-gravitational-energy} to
the non-time-symmetric case is under study.

\subsubsection{QLM at a MOTS without time symmetry}
\label{sub:QLM-at-MOTS}

In this section, we will numerically determine the QLM on CES families near
$\Sout$.
The goal is to justify the definition \eqref{eq:numerical-QLM-on-MOTS}
of the QLM on a MOTS $\Surf$ by
exploring its behavior as we approach $\Surf$ along the family
from the $\Theta_{+} > 0$ side, where the QLM is well defined.

Although $\rho$ \eqnref{eq:mass_density} and $\QLM$ \eqref{eq:QLM} seem well-defined even for $|H|\to 0$, the OEE that determines $\tau$ is clearly singular at a MOTS: $|H|$ appears as denominator in both \eqref{eq:j_flux} and \eqref{eq:alpH-TpTm}. Therefore, the definition of the QLM cannot trivially be extended to the case of a MOTS.
In other words, as $\Theta_+\to 0$, $|H|\to 0$, and so the assumption of having a spacelike
mean curvature vector in the Wang-Yau QLM breaks down.
We focus on examining this issue here numerically.
A mathematical treatment of this case will be given in Sec.~\ref{sec:proofs}.

\begin{figure*}\centering
    \includegraphics[clip,trim=0 20 0 0,width=0.8\linewidth,valign=t]{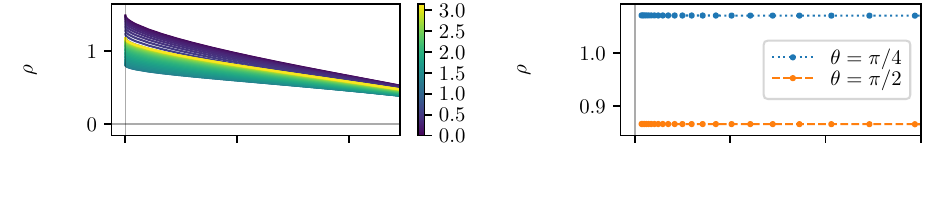}\\%
    \includegraphics[clip,trim=0 20 0 0,width=0.8\linewidth,valign=t]{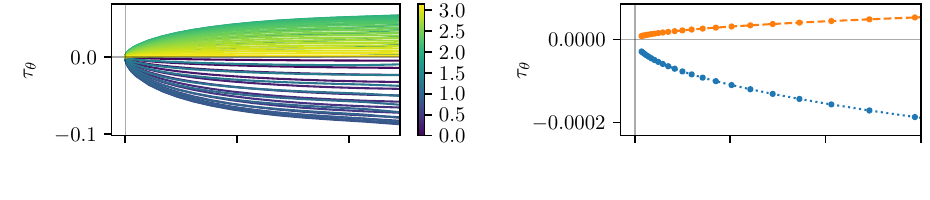}\\%
    \includegraphics[clip,trim=0 20 0 0,width=0.8\linewidth,valign=t]{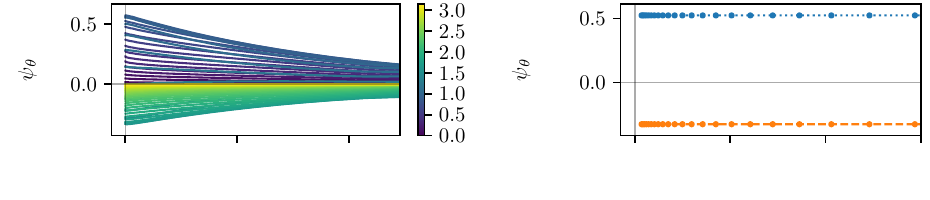}\\%
    \includegraphics[clip,trim=0 20 0 0,width=0.8\linewidth,valign=t]{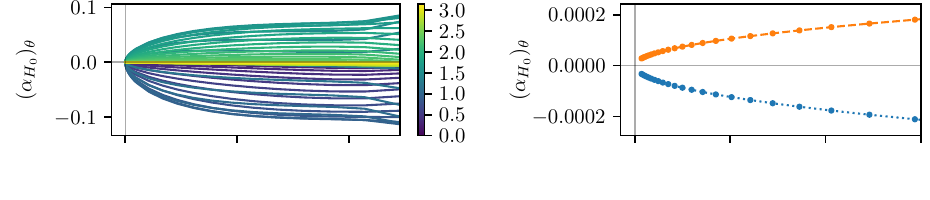}\\ %
    \includegraphics[clip,trim=0 20 0 0,width=0.8\linewidth,valign=t]{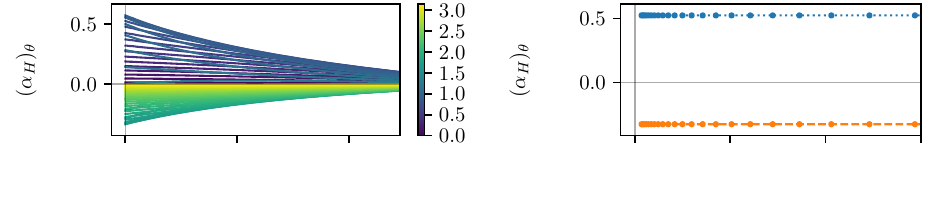}\\%
    \includegraphics[clip,trim=0  0 0 0,width=0.8\linewidth,valign=t]{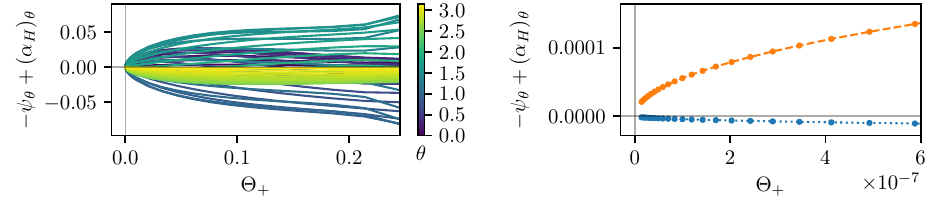}%
    \caption[]{\label{fig:BL31-MOTS-limit}%
        Terms appearing in \eqref{eq:j_flux}, i.e.,
        $j_\theta = \rho \tau_\theta - \psi_\theta - (\alpha_{H_0})_\theta + (\alpha_H)_\theta$,
        where
        $\psi \defeq \sinh^{-1}\left( \frac{\rho\Delta\tau}{|H_0|\,|H|} \right)$.
        They are plotted as function of $\Theta_{+}$ for CESs outside $\Sout$.
        For each curve, we fix the coordinate $\theta \in (0,\pi)$ along the
        surface and show its value as color.
        The right column is a close-up of $\Theta_{+}=0$ and depicts two
        individual values of $\theta$.
        The first five rows show the individual quantities whereas the final
        row shows $-\psi_\theta + (\alpha_H)_\theta$.
        The slice is from {\SimA} at time $t=2.5\,M$.
    }
\end{figure*}

\begin{figure}\centering
    \includegraphics[width=0.9\linewidth,valign=t]{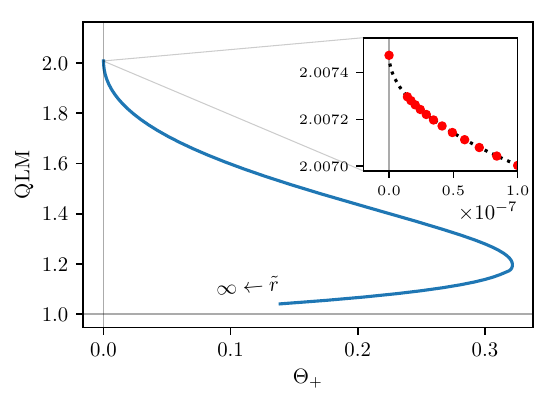}%
    \caption[]{\label{fig:BL31-QLM-extrapolate-MOTS}%
        QLM for CESs with $\Theta_{+}>0$ outside the MOTS $\Sout$.
        The value for $\Sout$ at $\Theta_{+} = 0$ is calculated using
        \eqref{eq:numerical-QLM-on-MOTS}.
        The dotted line in the inset interpolates the data points (red)
        and is obtained by fitting $c_1 - c_2 \sqrt{\Theta_{+}}$.
        The curve ends before reaching a QLM of $1$
        since the numerical slice does not extend to infinity.
        This is for simulation {\SimA} at time $t=2.5\,M$.
    }
\end{figure}

To numerically explore what happens as $\Theta_{+} \to 0$, we look at the
individual terms in the OEE \eqref{eq:j_flux} that determines $\tau$.
As can be seen in Fig.~\ref{fig:BL31-MOTS-limit},
$\rho$ remains finite while $\tau_\theta \defeq \tau'$ approaches zero, so
their product vanishes at a MOTS.
Since $\tau' \to 0$, $\tau\to \const$ and clearly $\alpha_{H_0}\to 0$.
The remaining terms $\psi_\theta$ and $\alpha_H$ both remain finite in the limit.
However, within numerical limits, they cancel in $j_\theta$ as $\Theta_{+}\to0$. The net result is thus $\tau\to 0$ is a solution to the OEE at a MOTS.

This suggests that as we approach a MOTS, $\tau$ becomes constant and $j_a$
vanishes, so that $\rho \to |H_0|$.
In terms of the QLM, this limit is shown in
Fig.~\ref{fig:BL31-QLM-extrapolate-MOTS} together with the value of the QLM
calculated at the MOTS using \eqref{eq:numerical-QLM-on-MOTS}.

\subsection{Time evolution of the QLM at a MOTS}
\label{sub:evolution}

Having argued that one can extend the Wang-Yau QLM to the common apparent
horizon and each individual horizon by \eqref{eq:numerical-QLM-on-MOTS}, we now
examine its time evolution during the head-on merger of two black holes.

If one interprets the Wang-Yau QLM as a measure to separate quasi-local
degrees of freedom from traveling gravitational waves, then
Fig.~\ref{fig:BL31-QLM-MOTS-AH-top-bot-evolution} indicates the region surrounded
by $\Sout$ loses energy/mass while sub-regions surrounded
by each individual horizon $\Sonetwo$ gain energy/mass during the collision. Furthermore, in the longer simulation {\SimB}, we find an oscillation in energy/mass contained inside the apparent horizon $\Sout$
(Fig.~\ref{fig:BL5-QLM-MOTS-evolution}). It is well known that for two black hole collisions, the intrinsic geometry of the apparent horizon experiences oscillations at the (lowest) quasi-normal frequency of the final black hole \cite{anninos:1994PRD}. While the apparent horizon area monotonically increases despite the oscillation at the quasi-normal frequency, $\QLM$ of the apparent horizon fails to maintain monotonicity with the oscillation. Nevertheless, as the final black hole settles down to equilibrium, the measure of $\QLM$ and area for the apparent horizon converges.

First of all, although the Wang-Yau quasi-local mass is defined for a 2-surface and does not depend on the choice of slicing $\Slice$, determining the apparent horizon through $\Theta_+$ does depend on the choice of slicing $\Slice$.
Bearing in mind this ambiguity associated with the apparent horizon, one tends to conclude that the Wang-Yau quasi-local mass suggests that the region bounded by the outermost MOTS $\Sout$
could lose energy to infinity.
This is a different picture from that indicated by the area of the horizons, which increase monotonically, and the standard first law of black hole thermodynamics. This difference is actually expected: equation (6.20) in \cite{Brown:1992br} indicates that the Brown-York quasi-local mass for a Schwarzschild black hole satisfies a balance equation that differs from the standard black hole thermodynamics. This again emphasizes that quasi-local mass as defined by Brown and York or Wang and Yau measures a different quantity than irreducible mass.
More specifically, as the horizon expands, more negative gravitational field energy is taken into account by $\QLM$ which counteracts the positive contribution from growing area and absorbed gravitational waves. This issue is crucial in understanding Wang-Yau quasi-local mass and will be carefully studied with a direct calculation of gravitational wave energy in a future study.

\begin{figure}\centering
    \includegraphics[width=0.9\linewidth,valign=t]{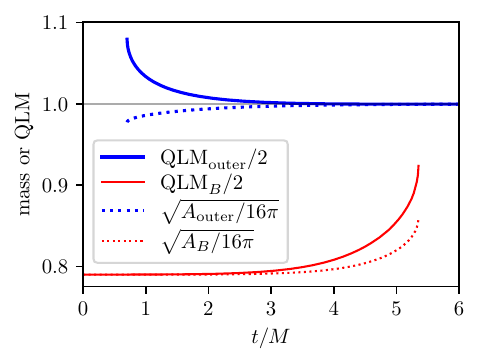}%
    \caption[]{\label{fig:BL31-QLM-MOTS-AH-top-bot-evolution}%
        QLM and area masses of $\Sout$ and $\Stwo$
        calculated using \eqref{eq:numerical-QLM-on-MOTS}
        in simulation {\SimA} as function of simulation time.
        The smaller horizon $\Sone$ (not shown) has a qualitatively similar
        behavior, though less pronounced, as $\Stwo$.
        For easier comparison, the QLM has been divided by $2$ to account for
        the fact that for Schwarzschild, the QLM of the horizon is twice the
        ADM mass.
    }
\end{figure}

\begin{figure}
    \includegraphics[width=1.0\linewidth,valign=t]{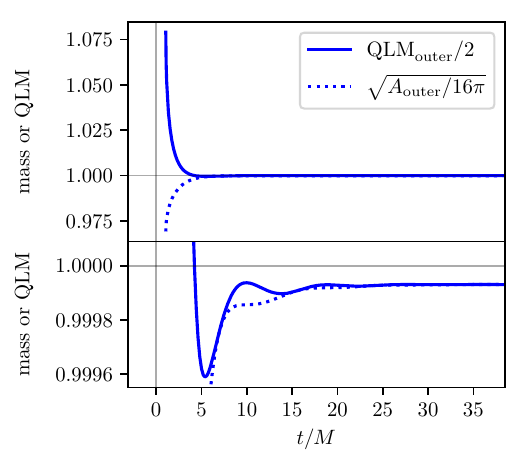}
    \caption[]{\label{fig:BL5-QLM-MOTS-evolution}%
        Evolution of the QLM calculated using \eqref{eq:numerical-QLM-on-MOTS}
        on $\Sout$ and the horizon's irreducible mass.
        This plot shows the longer simulation {\SimB}.
        At the final time,
        we have $\QLM_\text{outer}/2 \approx 0.9999303$
        and $\sqrt{A_\text{outer}/16\pi} \approx 0.9999309$.
        The lower panel is a close-up on the $y$-axis of the first panel.
    }
\end{figure}

\subsubsection{Investigating the hoop-conjecture}

A viable notion of QLM in general relativity should lead to numerous
applications.  We conclude this section on numerical results by
exploring one such application, namely the hoop conjecture
\cite{Thorne-Hoop,Senovilla:2007dw}.  This conjecture addresses the
question of under what conditions a black hole forms.  As we shall
shortly see, several aspects of the conjecture are not precisely
formulated and numerical relativity has a role to play in exploring
and evaluating various possibilities; see e.g. \cite{Ames:2023akp} for
recent work in this direction.

In the case of gravitational collapse, the intuitive picture is that
of matter fields getting compressed due to their self-gravity and
eventually becoming sufficiently dense to form a black hole.  As
originally formulated by Kip Thorne: Horizons form when and only when
a mass $M$ gets compacted into a region whose circumference
$\mathcal{C}$ in every direction does not exceed $4\pi GM/c^2$.  Thus,
in units with $G=c=1$, a horizon should form when, and only when,
$\mathcal{C}/4\pi M \lesssim 1$.  Here the notion of what one means by
mass is left vague, as is the space of curves (``hoops'') one should
use. Moreover, the value of 1 on the right hand side is motivated by
the Schwarzschild metric, and other numerical values might be
appropriate in general situations.

If a notion of QLM is generally viable, it should be possible to use
it as the appropriate mass within the hoop conjecture.  For the
Schwarzschild spacetime, as we have seen, the Wang-Yau QLM for the
horizon is twice the ADM mass. Thus, one might expect the relevant
hoop conjecture inequality should be modified to
$\mathcal{C}/4\pi M < 0.5$.  For the hoops, we shall use closed
geodesics lying within the constant expansion surfaces that we have
already found.  In our present case, we do not deal with gravitational
collapse, but rather with a binary black hole merger where we always
have black holes present on any time slice. We instead seek to
investigate the issue of when the common horizon forms, and whether
its formation can be predicted by a hoop conjecture argument using the
Wang-Yau QLM.  We assume further that the hoop conjecture applies to
the formation of marginally trapped surfaces.  In our case, the
constant expansion surfaces and the marginally trapped surfaces turn
out to be prolate so that the polar circumference $\mathcal{C}_{p}$ is
larger than the equatorial circumference.  Thus we calculate the ratio
$\mathcal{C}_{p}/4\pi M$ for constant expansion surfaces on different
time slices in the vicinity of the time when the common horizon is
first formed.

\begin{figure}
    \includegraphics[width=1.0\linewidth,valign=t]{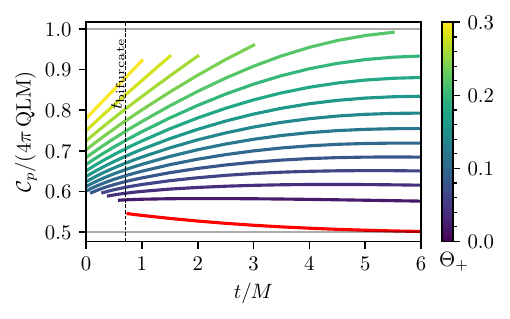}
    \caption[]{\label{fig:hoop}%
      Exploring the use of the Wang-Yau QLM in the hoop conjecture for
      the formation of the common horizon.  The figure shows the ratio
      $\mathcal{C}_{p}/(4\pi \times \mathrm{QLM})$ as functions of time
      for various constant expansion surfaces.  The curve in red
      refers to the common horizon while the other curves refer to
      positive values of the expansion; the color scale indicates the
      values of the expansion, which lie apart by $0.02$
      at the color scale's ticks.
      The time labeled as $t_{\rm bifurcate}$
      is when the common horizon is formed. }
\end{figure}
Our results are shown in Fig.~\ref{fig:hoop}. We see that the ratio
$\mathcal{C}_{p}/4\pi M$ approaches $0.5$ asymptotically as expected.
At earlier times, this ratio is somewhat larger than the limiting
value 0.5.  This is not unusual -- similar results were found in e.g.\
\cite{Ames:2023akp} where the ratio was about 12\% larger than the
limiting value predicted by the hoop conjecture. At each time just
before the horizon is formed, the value of $\mathcal{C}_p/4\pi M$
approaches the limiting value as the expansion becomes smaller.  These
results are suggestive but inconclusive -- it is not yet clear whether
this can be used as a prediction for the formation of the common
horizon.  This would require us to identify a suitable threshold for
$\mathcal{C}_p/4\pi M$ for these surfaces.  In the above discussion we
have considered only the constant expansion surfaces.  It is plausible
that there exist other surfaces which have a smaller value of
$\mathcal{C}_p/4\pi M$.  Therefore, while we shall not do so here, it
would be more satisfactory to consider a more general class of
2-spheres and to minimize the ratio $\mathcal{C}_{p}/4\pi M$ over
these spheres.  Following the spirit of the hoop conjecture, one could
then look for a threshold on this minimum value of
$\mathcal{C}_{p}/4\pi M$, and investigate whether it can predict the
formation of a black hole horizon.

\section{Defining the QLM on a MOTS}
\label{sec:proofs}

The above numerical results have already used a ``working definition''
\eqref{eq:numerical-QLM-on-MOTS} for evaluating the Wang-Yau QLM on a MOTS.
In this section, we explore the limit $\Theta_{+}\to0$ from a mathematical perspective to
justify this definition.
It seems plausible that an extension to the case with angular momentum is
possible.
However, this will be left to future work.

\begin{lemma}
Consider an axisymmetric collision with no angular momentum involved. Further assume isometric embedding or time function $\tau$ being axisymmetric, then $j\equiv 0$ when it is well-defined, i.e.\ when the mean curvature vector $H$ is spacelike.
\end{lemma}
\begin{proof}
Use the definition \eqref{eq:j_flux}. We first observe that under the above assumptions, $j_\phi\equiv 0$. It is easy to see that $\nabla _{\partial\phi}\tau =\nabla_\phi \bigg(\sinh^{-1} \frac{\rho \Delta \tau}{|H_0||H|}\bigg)=0$ by assumption. That connection one-forms vanish can be seen from computing Christoffel symbols, using subscripts $3,4$ for the $H$ and $J$ direction, 
\begin{align*}
    \Gamma_{\phi 3}^4
        &= \frac{1}{2}g^{4\mu}\Big(\partial_\phi(g_{3\mu})+\partial_3(g_{\phi \mu})-\partial_\mu (g_{\phi 3})\Big)\\
        &= \frac{1}{2} \,g^{4\mu}\, \partial_3(g_{\phi \mu})
\end{align*}
and noting that the spacetime under consideration possesses a symmetry $\phi\to -\phi$ such that cross terms in the metric involving $\phi$ all vanish. 

Next we show $j_\theta=0$. The most general axisymmetric metric for a 2-surface is
$$d\sigma^2 = r_0^2 P(r_0,\theta)^2 d\theta^2 + r_0^2 Q(r_0,\theta)^2 \sin^2\theta d\phi^2 \;.$$
Then 
\begin{align*}
    &\nabla_a j^a =
    \frac{1}{\sqrt{\sigma}}\partial_\theta (\sqrt{\sigma}\sigma^{\theta\theta}j_\theta)=   \frac{1}{\sqrt{\sigma}}\partial_\theta (\frac{Q}{P}\sin\theta\, j_\theta)=0\\
    &\frac{Q}{P}\sin\theta\, j_\theta = \const
    \;.
\end{align*}
Using that $Q, P$ do not vanish (metric non-degenerate), $\sin(0)=\sin(\pi)=0$ and that $j$ is well-defined, one gets $\const=0$ and hence $j_\theta\equiv 0$ on the surface.

\end{proof}
\begin{rmk}
    This may sound contradictory to the well-known fact that there are infinitely many non-vanishing, smooth, divergence-free vector fields on $S^2$. We elaborate on this. Denote the volume form by $\omega=\sqrt{\sigma} d\theta \wedge d\phi$. Take divergence-free vector field $j=j^\theta \partial_\theta$ with $j^\phi=0$. Then 
\begin{align*}
 \Div(j)\omega = \Lie_j\omega = d \iota_j \omega =0
 \;.
\end{align*}
Given that $H^1_{dR}(S^2) =0$, there exist a smooth function $f(\theta,\phi)$ such that $\iota_j \omega=df$. That is,
\begin{align*}
    r_0^2 PQ\sin\theta\,j^\theta d\phi=df=\frac{\partial f}{\partial \theta} d\theta + \frac{\partial f}{\partial \phi} d\phi 
    \;.
\end{align*}
It follows that $\frac{\partial f}{\partial \theta}=0$. Then noting that the RHS is only a function of $\phi$ while the LHS clearly has a dependence on $\theta$, it follows that $j^\theta=f=0$. We again reach $j\equiv 0$ when axisymmetry is imposed. 
\end{rmk}
\begin{rmk}
When angular momentum is present, the metric would have cross terms involving $\phi$ and the connection one-form $\alpha_H$ would generally not vanish ($\alpha_{H_0}=0$ since the reference spacetime is static). In fact, the quasi-local angular momentum as defined in \cite{ChenWangYau:2015cmp} vanishes exactly when $j^\phi=0$, consistent with our proof above. Recall that quasi-local angular momentum is defined as
$$E(\Surf,X,T_0,K=\partial_\phi)=-\frac{1}{8\pi}\int\,\langle K,T_0\rangle \rho + K^a j_a \;.$$
If one assumes $\Surf$ and $\tau$ both axisymmetric, $\langle K, T_0\rangle=0$. $E(\Surf,X,T_0,K=\partial_\phi)=0$ if and only if $j_\phi=0$.
\end{rmk}

\begin{theorem}
    Assume $H_0$ remains spacelike as $H$ turns into null or $|H|\to 0$, then the solution $\tau$ to the OEE approaches a constant as the surface approaches the apparent horizon. 
\end{theorem}  

\begin{proof}
Note that $H_0 = \hat{k} \check{e}_3 + \frac{\Delta\tau}{\sqrt{1+|\nabla \tau|^2}} \check{e}_4$,
where $\hat{k}$ is the mean curvature of the projected surface $\pi(i_0(\Sigma))$ in $\R^3$, with $\pi:\R^{3,1}\to \R$ being projection along $T_0$. Assuming $H_0$ remains spacelike thus implies that $\tau$ is well behaved since otherwise $\frac{\Delta\tau}{|\nabla\tau|}\to \infty$ and $H_0$ is surely timelike.

Having proved that $j\equiv 0$ as $|H|\to 0$, we next examine \eqref{eq:j_flux} term by term to show that $\nabla_a \big(\sinh^{-1} \frac{\rho\Delta \tau}{|H||H_0|}\big)$
needs to be bounded, which in turn leads to $\Delta\tau\to 0$ and hence $\tau \to \const$ as $|H|\to 0$.

First look at the $\alpha_H$ term. Use \eqref{eq:alpH-TpTm} and the fact that the family of $\Surf$ is constant expansion surface with $\partial_a \Theta_+ \equiv 0$,
\begin{equation*}
    \begin{split}
        \alpha_H 
        =\frac{1}{2} \left(-\frac{\partial_a \Theta_-}{\Theta_-}-\langle ^{(4)}\nabla_a \ell_-, \ell_+\rangle \right)
    \end{split}
\end{equation*}
and hence bounded as $\theta_+ \to 0$.

Next look at the $\alpha_{H_0}$ term. Using $\check{e}_4=\frac{T_0+\nabla \tau}{\sqrt{1+|\nabla \tau|^2}}$ and $\langle ^{(4)}\nabla_a \check{e}_3, T_0 \rangle=0$
$$\alpha_{H_0}(e_a) = \langle ^{(3,1)}\nabla_a \check{e}_3, \check{e}_4 \rangle = \Romannum{2}\left(e_a,\frac{\nabla\tau}{\sqrt{1+|\nabla \tau|^2}}\right)\;,$$
where $\Romannum{2}$ is the second fundamental form of the cylinder spanned by $i_0(\Surf)$ and $T_0$  in $\R^{3,1}$. Therefore $\alpha_{H_0}$ is also bounded as $\theta_+ \to 0$.

Lastly, note that
\begin{align*}
    \rho \nabla \tau
        =& \left(\sqrt{|H_0|^2+\frac{(\Delta \tau)^2}{1+|\nabla \tau|^2}}
            -\sqrt{|H|^2+\frac{(\Delta \tau)^2}{1+|\nabla \tau|^2}}\,\right)
        \\&\qquad\times
        \frac{\nabla\tau}{\sqrt{1+|\nabla\tau}|^2}
\end{align*}
also remains bounded as $\theta_+ \to 0$. 

Putting together, one sees that $\nabla_\theta \left(\sinh^{-1} \frac{\rho\Delta \tau}{|H||H_0|}\right)$ has to remain finite for $j\equiv 0$ as $|H|\to 0$. This is only possible when $\Delta\tau \to 0$ as $|H|\to 0$, i.e.\ $\tau \to \const$. 

To gain more understanding about the limit $\tau\to \const$,
we use another formula for $j$ that does not invoke picking the specific frame $\{e_H,e_J\}$ 
$$j=\rho\nabla\tau -\alpha_{\check{e}_3}+\alpha_{\overline{e}_3}\;.$$
Indeed,
\begin{align*}
    -\alpha_{H_0}+\alpha_{H} &- \nabla  \left(\sinh^{-1} \frac{\rho\Delta \tau}{|H||H_0|}\right)\\
        &= -\alpha_{H_0}-\nabla \psi_0 +\alpha_H +\nabla \psi\\
        &= -\alpha_{\check{e}_3}+\alpha_{\overline{e}_3}
        \;.
\end{align*}

Imposing $\tau=\const$, $\rho\nabla\tau=0$
and $\alpha_{\check{e}_3} = \langle ^{(3,1)}\nabla  \check{e}_3, \check{e}_4 \rangle=0$ , recalling that $\check{e}_4=\frac{T_0+\nabla \tau}{\sqrt{1+|\nabla \tau|^2}}$.
One is left with $\alpha_{\overline{e}_3}$ only.
Recall that $\overline{e}_3$ is the `spacelike' unit normal chosen by the gauge condition $$\langle \overline{e}_4, H\rangle=\langle \check{e}_4, H_0\rangle\;,$$
which vanishes for $\tau=\const$ or $i_0(\Surf)\subset \R^3$. When $H$ is spacelike, this gauge condition picks $\overline{e}_3\propto H$ and $\alpha_{\overline{e}_3}=\alpha_H$ which in general do not vanish. So $\tau=\const$ does not solve OEE in general.

Now consider the limit $\Theta_+\to 0$.
Since $$H = \frac{\Theta_+ \ell_- + \Theta_- \ell_+ }{2}, \quad J = \frac{\Theta_+ \ell_- - \Theta_- \ell_+ }{2} $$
as $\Theta_+\to 0$, $H$ and $J$ both turn to null vectors along $\ell_+$. The gauge condition $\langle \overline{e}_4, H\rangle=0$ forces $\overline{e}_4,\overline{e}_3 \propto \ell_+$. Thus $\alpha_{\overline{e}_3}=\langle \nabla \overline{e}_3, \overline{e}_4\rangle =0$ and $j=0$ is satisfied.
\end{proof}

\begin{rmk}
The above proposition shows that one can extrapolate the definition of the QLM to $|H|=0$ with the optimal embedding being $\tau =\const$. In this case, i.e.\ $\tau=\const$ and $|H|=0$,
    \begin{align*}
        \QLM &=\frac{1}{8\pi}\int \rho\\
        &= \frac{1}{8\pi}\int \frac{\sqrt{|H_0|^2+\frac{(\Delta \tau)^2}{1+|\nabla \tau|^2}}-\sqrt{|H|^2+\frac{(\Delta \tau)^2}{1+|\nabla \tau|^2}}}{\sqrt{1+|\nabla\tau}|^2}\\
        &= \frac{1}{8\pi}\int |H_0|
        \;.
    \end{align*}
\end{rmk}

\begin{rmk}
    As proved in \cite{ChenWangYau:2014cmp}, a solution $\tau$ to the optimal embedding equation is a local minimum of Wang-Yau quasi-local energy if $$|H_{\tau}|>|H| >0 \;,$$ where $H_\tau$ is the mean curvature vector of the isometric embedding with time function $\tau$. The condition is satisfied for constant expansion surfaces with $\Theta_+>0$ (Fig.~\ref{fig:BL31-MOTS-limit-H}).
    Assume continuity, $\tau=\const$ is a local minimum of Wang-Yau quasi-local energy at the apparent horizon.
\end{rmk}

\begin{figure}\centering
    \includegraphics[width=1.0\linewidth,valign=t]{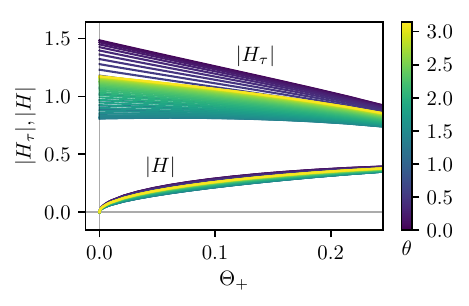}%
    \caption[]{\label{fig:BL31-MOTS-limit-H}%
        Value of $|H_\tau|$ and $|H|$ along the CESs outside
        $\Sout$ at simulation time $t=2.5\,M$ in {\SimA}.
        As Fig.~\ref{fig:BL31-MOTS-limit},
        the colors indicate the coordinates on the surface.
    }
\end{figure}
The numerical calculations in Sec.~\ref{sub:QLM-at-MOTS} comply with the above
results.

\section{Conclusions}
\label{sec:implications}

In this work, we have studied the Wang-Yau quasi-local mass during a
binary black hole merger.  The Wang-Yau QLM uses an embedding of
2-surfaces in Minkowski space. We have solved the optimal embedding
equation numerically and applied it to the head-on collision of two
non-spinning black holes starting with Brill-Lindquist initial data. We
numerically determined the QLM for surfaces close to the horizons and
for families of surfaces approaching the asymptotically flat ends and
study their time evolution, and also present a preliminary
investigation of the hoop conjecture applied to the formation of the
common horizon.  Finally, we have suggested an extension of the
Wang-Yau QLM to marginally trapped surfaces.

For a Schwarzschild black hole, our calculations agree with the well
known result of the Brown-York mass, i.e.\ the QLM is twice the ADM
mass on the horizon.
The Brown-York mass decreases monotonically as one moves outward from
the horizon, and for the sphere at infinity, it yields the ADM
mass. This is in sharp contrast to other quasi-local mass definitions
such as Hawking and Bartnik masses. The Wang-Yau quasi-local mass also
inherits this monotonic decreasing property. Such monotonicity is
clearly demonstrated numerically for our family of constant expansion
surfaces. Therefore, it is a crucial property of the Wang-Yau
quasi-local mass that some measure of negative gravitational field
energy is accounted for. In particular, for surfaces $\Surf$ in a
time-symmetric slice $\Slice$, an explicit expression for
gravitational field energy was written down
\cite{MiaoShiTam:2010CMP}. An analogous expression for the case of
non-time-symmetric slices is under study.

We have extended the definition of Wang-Yau quasi-local mass for
2-surfaces of spacelike mean curvature vector to 2-surfaces of null
mean curvature vector, i.e.\ \eqref{eq:numerical-QLM-on-MOTS}. With
this extension, we examined the time evolution of QLM for a black
hole, defined as QLM at $\Sout$, during the head-on collision of two
non-spinning black holes. As is well known, the area increases
monotonically throughout the evolution. At late stages, the area
evolution exhibits damped oscillations which are known to be
associated with the quasi-normal modes of the remnant black hole (see
e.g. \cite{Mourier:2020mwa}). In contrast, $\QLM$ decreases
monotonically at first and starts to lose monotonicity when
oscillations take place. We see from the bottom panel of
Fig.~\ref{fig:BL5-QLM-MOTS-evolution} that the oscillation frequency
of the QLM is similar to that of the area and one might expect these
to also be associated with the quasi-normal modes.  Monotonicity could
be included as one of the additional requirements for a QLM, and in
future work we shall investigate the possibility to modifying the
definition of the QLM appropriately to make it monotonic.

That QLM and area evolve differently seems to comply with a variation equation for the Brown-York mass in the Schwarzschild black hole case \cite{Brown:1992br}, which differs from the standard first law of black hole thermodynamics. One might again invoke negative gravitational field energy to explain this difference. Nevertheless, both of these measures, namely the area and half of $\QLM$ of the apparent horizon converge to the same value asymptotically as the final black hole settles down to its equilibrium state. We expect that employing estimates of gravitational wave energy will greatly clarify on various distinctive features
of the Wang-Yau quasi-local energy.

Future work will extend this study in various directions. Further extension of quasi-local definition for mass and angular momentum to surfaces of time-like mean curvature
vector is of great interest. An example is trapped surfaces inside the horizon, whose quasi-local mass might reveal important information about black holes. An attempt definition based
on Brown-York mass was made in a previous study \cite{LSYork:2007PRD},
where it is found the quasi-local mass either experiences an infinite
slope or a cusp at the horizon while the former was preferred by those
authors. A naive extension of Wang-Yau $\QLM$ presented here reveals a
similar behavior. A careful study of this issue will be presented
elsewhere.

Furthermore, it will be important to calculate the QLM in cases with rotation or
angular momentum. In these cases, we can also calculate the
quasi-local angular momentum. Quasiloal mass and angular momentum for surfaces inside Kerr's ergoregion might exhibit interesting features related to the Penrose process. The time variation of the mass and
angular momentum should be related to appropriate fluxes.  Moreover,
the event horizon, whose determination requires the knowledge of the
whole spacetime evolution, may reveal different information from the
apparent horizon studied here.  Finally, an extension of the QLM and angular
momentum to cosmological spacetimes would be of great interest.  Here
it would be necessary to consider a reference configuration not in Minkowski spacetime, but in de~Sitter or anti-de~Sitter spacetime \cite{ChenWangYau:2016dSAdS}.


\bibliography{references}

\end{document}